\documentclass[11pt,a4paper]{article}
\usepackage{amsmath, amsthm, amscd, amsfonts, amssymb, graphicx, color}
\usepackage[utf8]{inputenc}
\usepackage{graphicx}
\usepackage{url}
\usepackage[titlenumbered,ruled,noend,algosection]{algorithm2e}

\newtheorem{thm}{Theorem}[section]
\newtheorem{cor}[thm]{Corollary}
\newtheorem{lem}[thm]{Lemma}

\newtheorem{defn}[thm]{Definition}
\newtheorem{obs}[thm]{Observation}

\DontPrintSemicolon

\begin{document}
\title{Approximate Hotspots of Orthogonal Trajectories}
\author{Ali Gholami Rudi\thanks{
{Department of Electrical and Computer Engineering},
{Bobol Noshirvani University of Technology}, {Babol, Mazandaran, Iran}.
Email: {\tt gholamirudi@nit.ac.ir}.}}
\date{}
\maketitle

\begin{abstract}
We study the problem of finding hotspots,
i.e.~regions, in which a moving entity spends a significant amount of time,
for polygonal trajectories.
The fastest exact algorithm, due to Gudmundsson, van Kreveld, and Staals (2013)
finds an axis-parallel square hotspot of fixed side length in $O(n^2)$ for
a trajectory with $n$ edges.
Limiting ourselves to the case in which the entity moves in a
direction parallel either to the $x$ or to the $y$-axis,
we present an approximation algorithm with the time complexity $O(n \log^3 n)$
and approximation factor $1/2$.
\end{abstract}

\section{Introduction}
\label{sint}
Tracking technologies like GPS gather huge and growing collections
of trajectory data, for instance for cars, mobile devices, and animals.
The analysis of these collections poses many interesting
problems, which has been the subject of much attention recently \cite{zheng15}.
One of these problems is the identification of the region,
in which an entity spends a large amount of time.  Such regions
are usually called stay points, popular places, or hotspots in the
literature.

We focus on polygonal trajectories, in which
the trajectory is obtained by linearly interpolating
the location of the moving entity, recorded at specific
points in time (the assumption of polygonal trajectories
is very common in the literature; see for instance \cite{benkert10,buchin11,aronov16,rudi18}).
Gudmundsson et al.~\cite{gudmundsson13} define several problems about
trajectory hotspots and present an $O(n^2)$ exact algorithm to solve the
following: defining a hotspot as an axis-aligned square of fixed side length,
the goal is to find a placement of such a square that maximizes the time the
entity spends inside it for a trajectory with $n$ edges
(there are other models and assumptions about hotspots,
for a brief survey of which, the reader may consult \cite{gudmundsson13};
e.g.~the assumption of pre-defined potential regions \cite{alvares07},
counting only the number of visits or the number of visits
from different entities \cite{benkert10}, or
based on the sampled locations only \cite{tiwari13}).
To solve this problem, they first show that the function that
maps the location of the square to the duration the trajectory
spends inside it, is piecewise linear and its breakpoints happen
when a side of the square lies on a vertex, or a corner
of the square on an edge of the trajectory.
Based on this observation, they subdivide the plane into $O(n^2)$
faces and test each face for the square with the maximum duration.

In this paper, we limit ourselves to trajectories whose edges
are parallel to the axes of the coordinate system (we call
these trajectories \emph{orthogonal}).
One possible application of this problem is
finding regions in a (possibly multi-layer, 3-dimensional)
VLSI chip, with high wire density considering their current,
to identify potential chip hot spots.
The algorithm presented by Gudmundsson et al.~\cite{gudmundsson13}
finds an exact solution for this problem in $O(n^2)$;
we are not aware of a faster exact algorithm.
Our contribution is to provide a faster approximation algorithm
with constant approximation factor.
A $c$-approximate hotspot of a trajectory is a square,
in which the entity spends no less than $c$ times the time
it spends in the optimal hotspot.
We present an algorithm for this problem
with an approximation ratio of $1/2$ and the time complexity $O(n \log^3 n)$,
in which $n$ is the number of trajectory edges.
In this algorithm we combine kinetic tournaments \cite{basch99}
with segment trees to maintain the maximum among sums of a
set of piecewise linear functions.
We also present a simpler $O(n \log n)$ time algorithm
for finding $1/4$-approximate hotspots of orthogonal trajectories.

The rest of this paper is organized as follows.
In Section~\ref{sprel}, we introduce the notation used in this paper
and also present a simple $1/4$-approximation algorithm for
finding hotspots.
In Section~\ref{S12}, we present our main algorithm,
and finally, we extend the algorithm to trajectories
in $R^3$ in Section~\ref{scon}.

\section{Preliminaries and Basic Results}
\label{sprel}
A trajectory specifies the location of a moving entity through time.
Therefore, it can be described as a function that maps each
point in a specific time interval to a location in the plane.
Similar to Gudmundsson et al.~\cite{gudmundsson13},
we assume that trajectories are continuous and piecewise linear.
The location of the entity is recorded at different points in time,
which we call the vertices of a trajectory; when necessary we
use the notation $\tau (v)$ to denote the timestamp of vertex $v$.
We assume that
the entity moves in a straight line and at constant speed from
a vertex to the next (this simplifying assumption is very common in
the literature but there are other models for the movement
of the entity between vertices \cite{miller91}); we call the sub-trajectory
connecting two contiguous vertices, an edge of the trajectory.

We represent a trajectory with its set of edges.
This representation does not preserve the order of trajectory edges;
in the problem studied in this paper the order of trajectory
edges is insignificant and this representation is sufficient.
In \emph{orthogonal} trajectories, all trajectory edges are
parallel either to the $x$ or to the $y$-axis.
In \emph{horizontal} (similarly \emph{vertical}) trajectories all edges
are parallel to the $x$-axis ($y$-axis).
We assign a weight to each edge to show how long the entity was
moving on it (Definition~\ref{dweight}).

\begin{defn}
\label{dweight}
To an edge $e = uv$, we assign a weight $w_e$, which denotes
the duration of the sub-trajectory through its end points (the difference
between the time recorded for its end points),
therefore $w_e = \tau(v) - \tau(u)$, where $\tau(w)$ is the
timestamp of vertex $w$.
\end{defn}

Unless otherwise stated, we assume that all squares mentioned
in the rest of this paper to be axis-parallel and of side length $s$,
which is an input and fixed during the algorithms.

\begin{defn}
$\mathit{Square}(x, y)$ is a square whose lower left corner is
at position $(x, y)$ on the plane.
The weight of a square $r$ with respect to a trajectory $T$ is
the total duration in which the entity spends inside it.
We denote it with $w_T (r)$, or if there is no
confusion $w(r)$.  More formally, if the trajectory enters
square $r$ $m$ times, $w_T (r) = \sum_{i=1}^{m} (l_i - e_i)$,
in which $e_i$ and $l_i$ denote the time at which the entity
enters and leaves the square in its $i$-th visit, respectively.
\end{defn}

We now define two of the main concepts of this paper, i.e.~hotspots
and approximate hotspots (Definitions \ref{dhotspot} and \ref{dapprox}).

\begin{defn}
\label{dhotspot}
A \emph{hotspot} is a square with the maximum possible weight.
We denote the weight of a hotspot of trajectory $T$ with $h(T)$.
\end{defn}
\begin{defn}
\label{dapprox}
A $c$-approximate hotspot of a trajectory $T$, for $c \leq 1$,
is a square whose weight is at least $c$ times the weight of a hotspot of $T$.
\end{defn}

\begin{figure}
	\centering
	\includegraphics{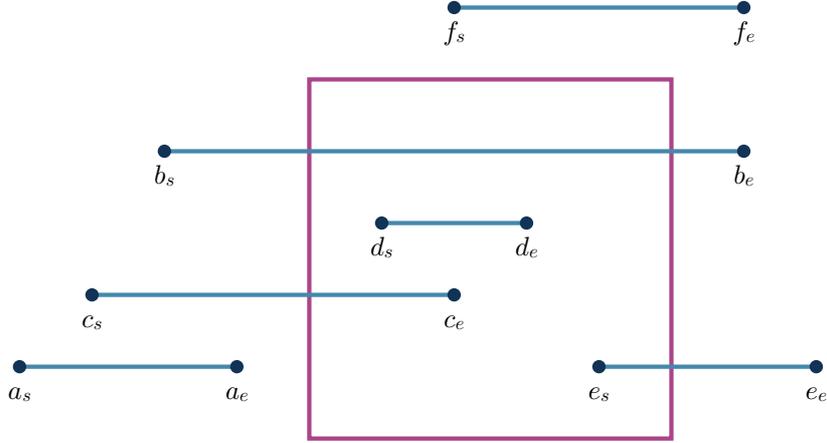}
	\caption{The contribution rate of $e$ is positive, the contribution rate of $c$
	is negative, and the contribution rate of every other edge is zero}
	\label{frate}
\end{figure}

\begin{defn}
\label{drate}
Let $r$ be an axis-parallel square and $H$ be a horizontal trajectory.
The contribution rate of an edge $e$ of $H$ with respect to $r$ is
the rate at which the contribution of the weight of the edge to the
weight of $r$ increases, if $r$ is moved to the right, in
the positive direction of the $x$-axis
(the contribution rate shows the slope of the curve in the
square position--edge contribution plane).
We denote the contribution rate of $e$ as $r(e)$.
\end{defn}
It is not difficult to see that the absolute value of $r(e)$
is either zero (for non-intersecting edges) or the ratio of
its duration to its length.
Note that if a square is moved to the left, the rate at which
an edge $e$ contributes to the weight of the square is $-r(e)$.
In Figure~\ref{frate}, the contribution rate of all edges except
$e$ and $c$ are zero.

\begin{defn}
\label{dtrate}
The contribution rate of a horizontal trajectory $H$ with
respect to square $r$, denoted as $r(H)$,
is defined as the sum of the contribution rates of all edges of $H$.
\end{defn}

We now present some preliminary results about orthogonal trajectories.

\begin{lem}
\label{hdecomp}
Let $H$ and $V$ be a partition of an orthogonal trajectory $T$,
in which $H$ contains its horizontal and $V$ contains its
vertical edges.
Let $h$ be the maximum of $h(H)$ and $h(V)$.
Then, $h$ is at least $h(T) / 2$.
\end{lem}
\begin{proof}
Let $r$ be a hotspot in $T$.  Every edge of $T$ is either
in $H$ or in $V$ and thus $w_H (r) + w_V (r)$ equals $h(T)$.
Therefore, either $w_H (r) \ge h(T) / 2$ or $w_V(r) \ge h(T) / 2$.
Since $h(H) \ge w_H (r)$ and $h(V) \ge w_V (r)$,
we have $\max(h(H), h(V)) \ge h(T) / 2$ as required.
\end{proof}

\begin{lem}
\label{hside}
Let $H$ be a horizontal trajectory.  There exists at least one square,
whose weight equals $h(H)$
such that one of its vertical sides contains a vertex of $H$.
\end{lem}
\begin{proof}
Let $r$ be a square with weight $h(H)$ (and thus a hotspot)
and suppose none of its vertical sides contains a vertex of $H$.
Clearly, $r(H)$ cannot be positive; otherwise,
the weight of $r$ increases by moving it to the right, which
is impossible since it is a hotspot.
Similarly, $r(H)$ cannot be negative (otherwise, the
weight of $r$ increases by moving it to the left, which is
again impossible).
Therefore, $r(H)$ is zero and by moving $r$ to the right
until one of its sides meets a vertex of $H$, its weight
does not change.
\end{proof}

\begin{lem}
\label{Lcorner}
Let $H$ be a horizontal trajectory.
Among all squares with at least a corner coinciding with a
vertex of $H$, let $h$ be the weight of a square with the maximum
weight.  Then, $h \ge h(H) / 2$.
\end{lem}
\begin{proof}
Let $r$ be the square with weight $h(H)$, one of whose vertical
sides contains a vertex $v$ of $H$ (such a square surely exists, as
shown in Lemma~\ref{hside}).  Suppose $v$ is on the left side
of $r$ (the argument for the right side is similar).
Let $r'$ and $r''$ be the squares with side length $s$,
whose lower left and upper left corners are on $v$,
respectively.  Given that the union of $r'$ and $r''$ covers $r$,
$w(r') + w(r'')$ is at least $h(H)$ and therefore
$\max(w(r'), w(r''))$ is at least $h(H) / 2$.  Since
$h \ge \max(w(r'), w(r''))$, we have $h \ge h(H) / 2$.
\end{proof}

In Theorem~\ref{L14}, we show how to find a maximum weight square
with a left corner on a trajectory vertex, for horizontal
trajectories.
The algorithm sweeps the plane horizontally.
We call any square whose left side is on the sweep line, a sweep square.
In Definition~\ref{Dcontrib}, we define the contribution function
of an edge.

\begin{defn}
\label{Dcontrib}
The contribution function $c_{e_i}(x)$ of the $i$-th edge $e_i$
of a horizontal trajectory $H$ shows the contribution of
$e_i$ to its intersecting sweep squares, when the sweep
line is at position $x$ on the $x$-axis.
The contribution function of an edge is piecewise linear.
Let $c_{e_i}(x) = a_i x + b_i$, where $a_i$ is the slope and
$b_i$ is the vertical intercept of $c_{e_i}$.
Table~\ref{tsweep} and Figure~\ref{fsweep} show the definition
of this function based on the relative position of an edge to
the sweep line; edge $e_i$ in Figure~\ref{fsweep} corresponds
to case i of Table~\ref{tsweep}.
\end{defn}

\begin{table}
\caption{The contribution function of a horizontal edge $e_i$
to their intersecting sweep squares ($c_{e_i} (x) = a_i x + b_i$).
The sweep line is at $x$, the right side of sweep squares is at $x + s$,
the $x$-coordinate of the left end point of edge $e_i$ is $x_i$,
the $x$-coordinate of the right end point of edge $e_i$ is $x'_i$,
the duration of $e_i$ is $w_i$,
and the ratio of the duration of $e_i$ to its length is $m_i$ ($m_i = w_i / (x'_i - x_i)$).
}
\label{tsweep}
\centering
\begin{tabular}{|c|c|l|l|}
\hline
Case & Condition & $a_i$ & $b_i$\\
\hline
1 & $x_i \le x'_i \le x \le x+s$	& 0	& 0 \\
2 & $x_i \le x \le x'_i \le x+s$	& $-m_i$	& $m_i x_i - m_i w_i$\\
3 & $x_i \le x \le x+s \le x'_i$	& 0	& $m_i s$ \\
4 & $x \le x_i \le x'_i \le x+s$	& 0	& $w_i$ \\
5 & $x \le x_i \le x+s \le x'_i$	& $m_i$	& $m_i s - m_i x_i$ \\
6 & $x \le x+s \le x_i \le x'_i$	& 0	& 0 \\
\hline
\end{tabular}
\end{table}
\begin{figure}[ht]
\centering
\includegraphics{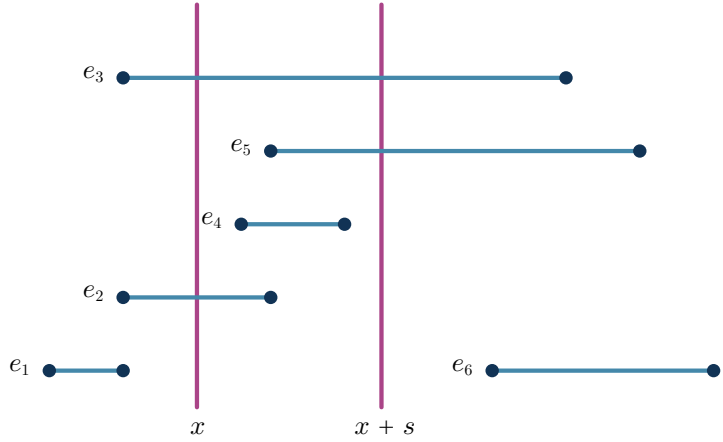}
\caption{
The left vertical line (at $x$) is the sweep line and
the right vertical line (at $x + s$) is the line at which the right
side of sweep squares lie.
Representing the $x$-coordinate of the left end point of
edge $e_i$ with $x_i$,
the duration of $e_i$ with $w_i$,
and the ratio of the duration of $e_i$ to its length with $m_i$,
we have
$c_{e_1}(x) = 0$,
$c_{e_2}(x) = -m_2 \cdot (x - x_2) + w_2$,
$c_{e_3}(x) = m_3 \cdot s$,
$c_{e_4}(x) = w_4$,
$c_{e_5}(x) = m_5 \cdot (x + s - x_5)$,
and $c_{e_6}(x) = 0$.
Edge $e_i$ corresponds to case i of Table~\ref{tsweep}.
}
\label{fsweep}
\end{figure}

\begin{thm}
\label{L14}
Among all axis-parallel squares with side length $s$ and with a
left corner on a vertex of a horizontal trajectory $H$ with $n$ edges,
the square with the maximum weight can be found with
the time complexity $O(n \log n)$.
\end{thm}
\begin{proof}
Let $\sigma = \left(e_1, e_2, ..., e_n\right)$ be
the sequence of trajectory edges, ordered by their $y$-coordinates.
We sweep the plane containing trajectory edges horizontally towards
the positive direction of the $x$-axis.
During the algorithm, for each edge $e_i$ we maintain its
contribution function, $c_{e_i}(x)$.

We can store the function assigned to trajectory edges in two Fenwick
trees \cite{fenwick96}: $A$ for storing the slope and
$B$ for storing the vertical intercept of the functions
in the order specified by $\sigma$.
More precisely, we store $a_i$ as the $i$-th element of $A$ and $b_i$ as
the $i$-th element of $B$.
During the algorithm we need to compute the sum of the functions of
a contiguous subsequence of the edges according to $\sigma$,
like $\sum_{q=i}^jc_{e_q}$.  To do so,
we find the sum of the elements
$i$ to $j$ in $A$ to obtain its slope and
the elements $i$ to $j$ in $B$ to obtain its vertical intercept in $O(\log n)$
using the Fenwick data structure.

The plane sweep algorithm processes four types of events:
when the left or right end point of an edge, which we
call an event point, meets the left or right side of sweep squares.
At each of these events for edge $e$, the slope and vertical
intercept of $c_e$ is updated to reflect the current
contribution of the edge to the weight of intersecting sweep squares.

To find a maximum weight square with a left corner on a
trajectory vertex (as required),
it suffices to compute the weight of sweep squares, when
one of their left corners coincides with an event point during
the algorithm.  We do this as follows.
For each event during the plane sweep algorithm for an edge
$e_i$, we first update the value of $c_{e_i}$; the slope
and the vertical intercept of the function of $e_i$ are updated
based on the relative position of the sweep line and the edge.
At an event for edge $e_i$,
the value of $a_i$ and $b_i$ are updated according to
Table~\ref{tsweep} in $A$ and $B$, respectively.

Let $k$ be the smallest index such that the vertical distance
between $e_k$ and $e_i$ (the difference between their $y$-coordinates)
is at most $s$.
Similarly, let $j$ be the largest index such that the
vertical distance between $e_i$ and $e_j$ is at most $s$;
$k$ and $j$ can be found using binary search on $\sigma$.
Then, the weight of the sweep square whose upper side is
at $e$ is $\sum_{q=j}^i c_{e_q}$ and the weight of the
sweep square whose lower side is at $e$ is $\sum_{q=i}^k c_{e_q}$.
These can be computed in $O(\log n)$ as mentioned before.

Therefore, when the algorithm finishes after processing $O(n)$
events, each with the time complexity $O(\log n)$, we can report
the maximum weight square with a left corner on a trajectory
vertex.
\end{proof}

\begin{thm}
\label{tapprox2}
There is an algorithm for finding an axis-parallel square
of side length $s$ for an orthogonal trajectory $T$,
such that the weight of the square found by the algorithm
is at least $1/4$ of the weight of a hotspot.
\end{thm}
\begin{proof}
Let $T$ be an orthogonal trajectory.
$T$ can be partitioned into sets $V$ and $H$
containing the vertical and horizontal edges of $T$, respectively.
Theorem~\ref{L14} shows how to find a square $r_H$ with the maximum
possible weight, in which one of its corners is on a vertex of
$H$ (the algorithm can be performed twice, once after rotating the
plane 180 degrees to find
the maximum-weight squares with one of its right corners on a vertex of $H$).
The same algorithm can obtain a square $r_V$ with the maximum possible weight
for $V$, after rotating the plane 90 degrees.
By Lemma~\ref{Lcorner}, $w(r_H) \ge h(H) / 2$ and $w(r_V) \ge h(V) / 2$.
Also, by Lemma~\ref{hdecomp}, $\max(h(H), h(V)) \ge h(T) / 2$,
implying that $\max(w(r_H), w(r_V)) \ge h(T) / 4$, as required.
\end{proof}

\section{A $1/2$ Approximation Algorithm}
\label{S12}
In the proof of Theorem~\ref{L14}, to find a hotspot of a horizontal trajectory
with $n$ edges, as we moved a vertical sweep line to the right,
we maintained the contribution of each edge to intersecting sweep squares
(squares whose left side is on the sweep line) and computed the weight of
sweep squares when necessary.
Instead, to improve the approximation factor of the algorithm,
in this section we maintain the weight of
sweep squares (not just edge contributions) during the algorithm.
Before presenting the details of the algorithm in Section~\ref{ssmain},
we provide an overview, and review kinetic tournament
trees and segment trees in Section~\ref{ssoverview}.

\subsection{Algorithm Overview}
\label{ssoverview}
The weight of a sweep square is the sum of the contributions
of trajectory edges to its weight, and thus, piecewise linear.
Although there are infinitely many sweep squares,
Lemma~\ref{ltrack} implies that
for finding a square with the maximum weight we can keep track
of only $n$ of them, that is exactly those whose upper side
has the same height ($y$-coordinate) as a trajectory vertex.

\begin{lem}
\label{ltrack}
Let $H$ be a horizontal trajectory and let $r$ be a square
with a non-zero weight.  There exists a square $r'$ such
that $w(r') \ge w(r)$ and the $y$-coordinate of a
vertex of $H$ is equal to the $y$-coordinate of the upper
side of $r'$.
\end{lem}
\begin{proof}
If a vertex of $H$ has the same $y$-coordinate as the
upper side of $r$, we are done.  Otherwise, we obtain
$r'$ by moving $r$ downwards until the first vertex of $H$
intersects the line containing the upper side of $r$.
Clearly, any part of any edge in $r$ is also
in $r'$, and therefore, $w(r') \ge w(r)$.
\end{proof}

Therefore, our goal in this section is to maintain the weight of $n$
specific squares as we move them with the sweep line in a plane sweep algorithm.
More precisely, we maintain the weight of $\mathit{Square}(x, y_i - s)$
for $1 \le i \le n$, where $y_i$ is the $y$-coordinate of the $i$-th
edge of the trajectory, as we move the sweep line ($x$ is variable
and denotes the position of the sweep line).

\begin{defn}
\label{dtrack}
For a horizontal trajectory $H$,
the $i$-th tracked square is $\mathit{Square(x, y_i - s)}$,
in which $y_i$ is the height of the $i$-th edge of the
trajectory and $x$ denotes the position of the sweep line.
The function $h_i\colon x \to w(\mathit{Square}(x, y_i - s))$
shows the weight of the $i$-th tracked square with respect
to $H$, when the sweep line is at $x$.
\end{defn}

We define a plane $P_H$ for horizontal trajectory $H$,
whose horizontal axis represents the position of the sweep
line and whose vertical axis represents the weight of tracked squares.
We add $n$ curves to $P_H$: the $i$-th curve is for the function $h_i$.
Suppose the highest point of the curves in $P_H$ is for curve $h_i$ at $x = a$.
Obviously, the square with the maximum weight is $\mathit{Square}(a, y_i - s)$.
Thus, to find a hotspot of the trajectory, we can compute
the upper envelope of $P_H$.
The plane $P_H$ is shown for an example horizontal trajectory
in Figure~\ref{Fenv}; the second tracked square at $x = a$ achieves
the maximum weight.

\begin{figure}
	\centering
	\includegraphics[width=\linewidth]{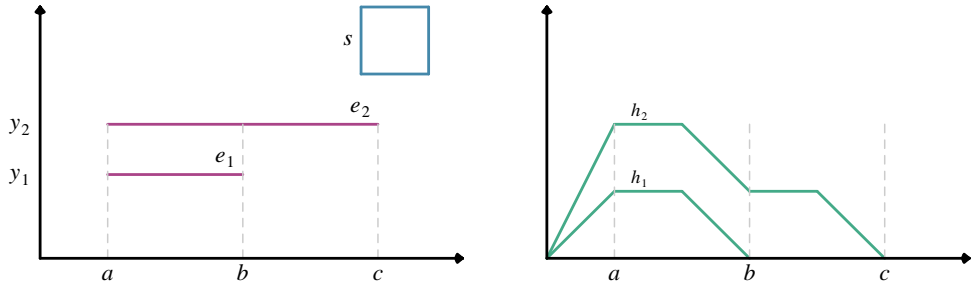}
	\caption{An example horizontal trajectory $H$ (left) and the plane $P_H$ (right) showing
	the weight of tracked squares}
	\label{Fenv}
\end{figure}

One solution for computing the upper envelope of $P_H$ is using kinetic
data structures.
To use the common terminology of kinetic data structures,
let the $x$-axis in $P_H$ denotes the time.
We want to maintain the point with the maximum height in any of the curves
as we move forward in time.
For this purpose, we can use kinetic tournament trees \cite{basch99},
which we shall briefly describe as follows.
We use an arbitrary binary tree with $n$ leaves, in each of the nodes of
which we store a function; let $f_v$ denote the function stored
at node $v$.
We initialize the tree for time $x = t_0$ as follows.
The function describing the $n$ curves of $P_H$ ($h_i$ for $1 \le i \le n$)
are stored in the leaves of the tree in an arbitrary order.
For each non-leaf node $v$, $f_v$ is recursively initialized as follows.
Let $u_1$ and $u_2$ be the children of $v$.
Without loss of generality, suppose $f_{u_1}(t_0) \ge f_{u_2}(t_0)$.
We set $f_v = f_{u_1}$ and call $f_{u_1}$ the winner at $v$ and
give it a winning certificate.  The winning certificate may not be
valid indefinitely.
Let $t_1$ be the earliest time in the future ($t_1 > t_0$),
at which $f_{u_2}(t_1)$ gets larger than $f_{u_1}(t_1)$.
We say that the winning certificate of $f_{u_1}$ fails at time $t_1$.
At this point, the winner at $v$ should be updated.
The certificate failure times (at which \emph{failure events} occur)
of all non-leaf nodes of the tree
are stored in a priority queue (the \emph{event queue}) and
are processed ordered by their time.
When the next certificate fails, the winner and the certificate failure
time of the corresponding node and its parent is updated.
Therefore, the function with the maximum value is always stored
at the root of the tree at any point of time.

To make the computation of certificate failure times more efficient,
we store a linear function at each node.
However, the functions $h_i$ (for $1 \le i \le n$) are piecewise linear
(the break point happen when a trajectory edge enters or leaves
$\mathit{Square}(x, y_i - s)$).
Consequently, the linear function assigned to the leaf corresponding to $h_i$
should be updated at the break points of $h_i$;
the time at which these updates should be performed (\emph{update events})
are also inserted into the event queue.
At each update event for curve $h_i$,
we need to update the function of the corresponding leaf; this
may change the winner and the failure time of the winning
certificate of the ancestors of the leaf.
The functions can be updated as in dynamic and kinetic tournament trees \cite{agarwal08}.

The main challenge that we try to address in this section is
reducing the cost of updating the tree.
An edge may intersect $\mathit{Square}(x, y_i - s)$ for $O(n)$
different values of $i$.  On the other hand, since there are $O(n)$ edges,
during the plane sweep algorithm we need to update the functions
of the leaves of the tree $O(n^2)$ times.
This makes the complexity of the plane sweep algorithm $\Omega(n^2)$.
To handle these updates more efficiently, we use a segment tree as
the underlying data structure for the kinetic tournament tree.
The details, correctness, and complexity of this algorithm is shown
in the rest of this section.

We close this section with a brief introduction to segment trees.
For a detailed introduction, the reader may consult classic
texts such as \cite{berg08}.
A segment tree is a balanced binary tree that can be used for answering
stabbing queries for a set of segments (or intervals) on a line.
It is initialized with a set of $n$ segments.
The start and end points of these segments split the line into
many \emph{elementary intervals}.
These intervals form the leaves of the segment tree in sorted order.
To each node $v$, an interval is assigned $\mathrm{Int}(v)$, representing
the union of the intervals of the leaves of the subtree rooted at $v$.
Each node of the segment tree $v$ also stores a subset of the input
segments $I(v)$.  An example segment tree is shown in
Figure~\ref{Fseg}.  In this figure the interval near each vertex $v$
is $\mathrm{Int}(v)$ and the set near each vertex is $I(v)$ (empty sets are
not shown).

\begin{figure}[ht]
	\centering
	\includegraphics[width=\linewidth]{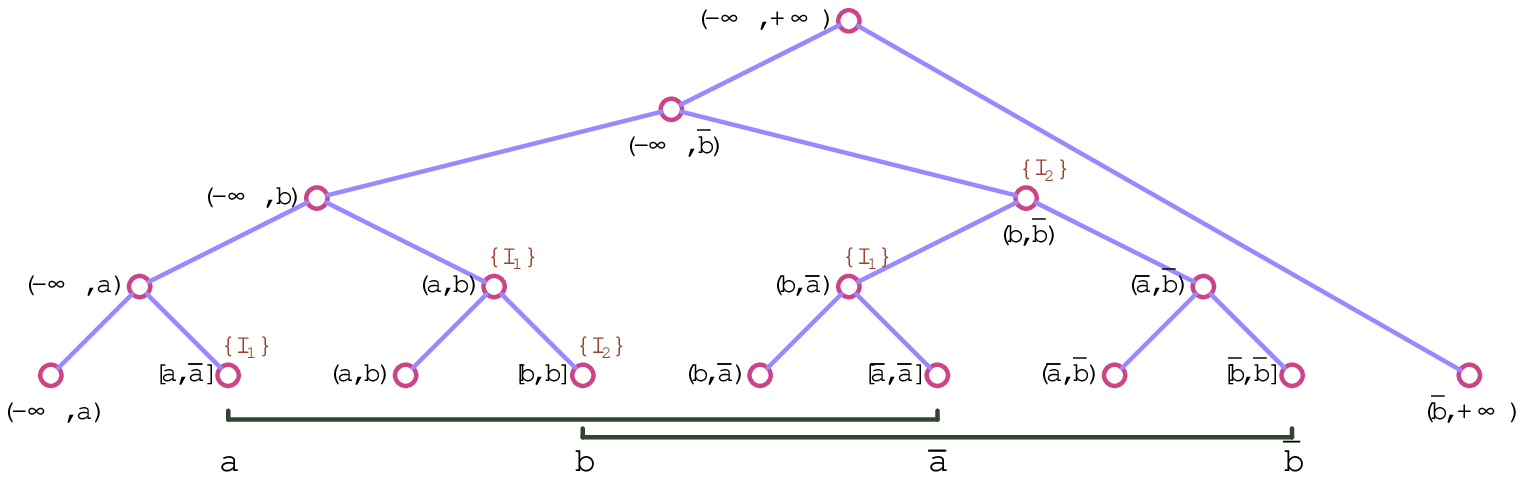}
	\caption{A segment tree for segments $I_1 = [a, \overline{a}]$ and $I_2 = [b, \overline{b}]$,
	where $a < b < \overline{a} < \overline{b}$}
	\label{Fseg}
\end{figure}

\begin{defn}
\label{Dsegleaf}
Let $[a, a']$ be a segment inserted into a segment tree.
In one of the leaves $v$ of the tree, we have
$\mathrm{Int}(v) = [a, a]$.  We call this leaf,
the segleaf of segment $[a, a']$.
\end{defn}

To answer a stabbing query for value $q$ (reporting every segment
containing the value $q$), the nodes of the tree on a path from the
root to a leaf are traversed and every segment in $I(v)$ for every
$v$ in this path is reported; every output segment is reported exactly
once and the complexity of answering a query is $O(\log n + k)$,
where the size of the output (the number of intervals containing $x$)
is $k$.

\subsection{The Algorithm}
\label{ssmain}
We use a segment tree $T$ in our plane sweep algorithm.
We assume that the input horizontal trajectory $H$
has $n$ edges and the $y$-coordinate of the $i$-th edge is $y_i$.
$T$ is initialized with $n$ segments: the $i$-th segment $g_i$
is $(y_i - s, y_i)$, corresponding to the $i$-th trajectory edge $e_i$.
A sweep square $\mathit{Square}(x, q)$ intersects $s_i$
during the plane sweep algorithm, if $q$ intersects $g_i$.
This yields the following key observation.

\begin{obs}
\label{Ostab}
A stabbing query for value $y = q$ on $T$ reports
every edge intersecting the sweep square $\mathit{Square}(x, q)$.
\end{obs}

An example is shown in Figure~\ref{Fseg2}.
For each edge, a segment of length $s$ is inserted into the
segment tree.
A stabbing query for $y = q$, yields the segments intersecting
the sweep square whose lower side has height $q$.

\begin{figure}[ht]
	\centering
	\includegraphics[width=\linewidth]{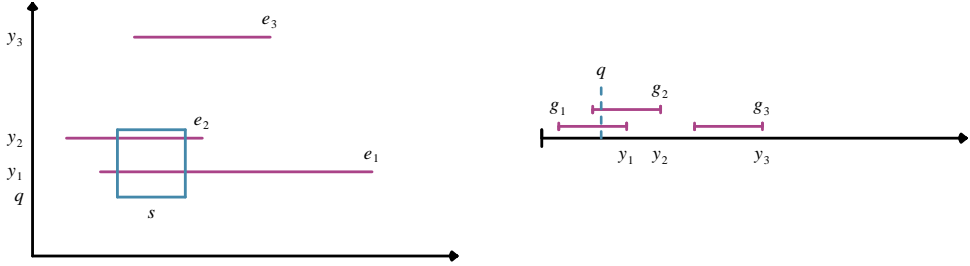}
	\caption{A horizontal trajectory (left) and the segments inserted
		into the segment tree (right)}
	\label{Fseg2}
\end{figure}
For each node $v$, as in regular segment trees, we represent the set of
segments stored at $v$ as $I(v)$ and the union of the intervals
of the leaves of the subtree rooted at $v$ as $\mathrm{Int}(v)$.
For a node $v$ of $T$, $v_y$ denotes the midpoint of $\mathrm{Int}(v)$.
With slight abuse of notation, for a leaf $v$ we sometimes
use $v$ to refer to $\mathit{Square}(x, v_y)$, and by the
weight of $v$, we mean $w(\mathit{Square}(x, v_y))$.

During the algorithm, we maintain the contribution function
of each edge (Definition~\ref{Dcontrib}) and update
it according to Table~\ref{tsweep}.
We sometimes use $c_{g_i}(x)$ to refer to
the contribution function of $e_i$, $c_{e_i}(x)$,
and refer to it as the contribution function of segment $g_i$.
We maintain the following attributes for each node $v$ of $T$
during the algorithm.

\begin{description}
\item[$s_v(x)$]:
The \emph{sum function} $s_v(x)$ is the sum of the functions assigned to
the segments in $I(v)$, i.e.~$s_v(x) = \sum_{g\in I(v)}c_g$.
Since $s_v$ is linear, it is enough to store its slope and vertical
intercept for each node $v$.
\item[$f_v(x)$]:
The \emph{winner function} $f_v(x)$ is equal to $s_v(x)$ for leaves and,
for other nodes, is the sum of $s_v(x)$ and the maximum
function $f_u(x)$ (for the current sweeping position $x$), among
any child $u$ of $v$.
Like $s_v$, $f_v$ is the sum of linear functions, and thus, linear.
\item[$\mathrm{winner}(v)$]:
The winner leaf $\mathrm{winner}(v)$ for a leaf node $v$
equals $v$ itself.  For other nodes,
$\mathrm{winner}(v)$ equals $\mathrm{winner}(u)$, if $u$ is the child
of $v$ with the maximum value of $f_u$ for the current sweep line
position (therefore, $f_v = f_u + s_v$).
\end{description}

Sweeping starts at $x = -\infty$, at which the functions of all
segments are zero.
We use a priority queue $Q$ to store sweeping events.  There are
two types of events: the failure of the winning certificate of
a node of $T$ (failure events) and a trajectory edge entering or
leaving tracked squares (update events).

At a failure event for node $v$, the winner, the winner function,
and the certificate failure time of $v$ and some of its ancestors
are updated, as in regular kinetic tournament trees.
At an update event for edge $e_i$, $c_{e_i}$ changes according
to Table~\ref{tsweep}, and the function $s_u$ for every $u$ such
that $g_i \in I(u)$ should be updated to reflect this change.
After updating $s_u$, $f_u$ also needs to be updated.
Since the updated function may change the winner and the failure
time of $u$'s parents, they should also be updated as in failure
events.

\begin{lem}
\label{Lleafweight}
For a leaf $v$, let $P = p_1 p_2 ... p_h$ be the path from
the root of $T$ to $v$, in which $p_h = v$.
Then, the weight of the square at $v$ is $\sum_{i=1}^{h} s_{p_i}$
during the plane sweep algorithm.
\end{lem}
\begin{proof}
Answering a stabbing query for a given value $y$,
requires a traversal from the root of the tree to a leaf and
reporting every segment stored in the nodes of the traversal.
Therefore, a query for the value of $v_y$ reports every segment
in $I(p_i)$ for $1 \le i \le h$.
To compute the weight of $v$, we need to sum up the contribution
of each intersecting segment.
Since, $s_{p_i}$ is the sum of the functions of the segments in $I(p_i)$
and the label of every intersecting segment is stored in exactly
one node of $P$, $\sum_{i=1}^{h} s_{p_i}$ is the total
contribution of the segments to the weight of the square.
\end{proof}

\begin{lem}
\label{Lleafwinner}
Let $v$ be a leaf, $u$ be one of its ancestors in $T$,
and $u_1 u_2 ... u_h$ be the path from $u$ to $v$,
where $u_1 = u$ and $u_h = v$.
If $\mathrm{winner}(u)$ is $v$, we should have $f_u = \sum_{i=1}^h s_{u_i}$.
\end{lem}
\begin{proof}
We use induction on $h$, the number of the nodes of the
path from $u$ to $v$.  When $h = 1$, $u$ is a leaf and $f_u = s_u$.
For $h > 1$, let $x$ and $y$ be the two children of $u$.
$\mathrm{winner}(u)$ is either $\mathrm{winner}(x)$ or $\mathrm{winner}(y)$.
Therefore, either $\mathrm{winner}(x)$ or $\mathrm{winner}(y)$ is $v$.
Without loss of generality, suppose $\mathrm{winner}(x) = v$ (this
implies that $u_2 = x$).  Based on induction hypothesis, $f_x = \sum_{i=2}^h s_{u_i}$.
Since $f_u = f_x + s_u$, the statement follows.
\end{proof}

\begin{lem}
\label{Lmaxsquares}
A square at one of the leaves of $T$ has the maximum weight
among all sweep squares at any stage during the algorithm.
\end{lem}
\begin{proof}
Lemma~\ref{ltrack} implies that by moving a square downwards
until its upper side has the same height as a trajectory
vertex, its weight cannot decrease.
Therefore, a tracked square $\mathit{Square}(x, y_i - s)$
for some index $i$ ($1 \le i \le n$) has the maximum weight.
The segleaf of $g_i$ (Definition~\ref{Dsegleaf}), whose
weight is $\mathit{Square}(x, y_i - s)$,
surely appears as a leaf of $T$, as required.
\end{proof}

\begin{thm}
\label{Talg12}
$T$ stores a sweep square with the maximum weight at its root
during the plane sweep algorithm.
\end{thm}
\begin{proof}
For the subtree rooted at any node $v$ of $T$ during the algorithm,
we show that $\mathrm{winner}(v)$ is the square with the
maximum weight among the squares, the height of whose lower side
is in the interval $\mathrm{Int}(v)$ (that is, we show that the
statement is true for every node and not just the root).
Lemma~\ref{Lmaxsquares} implies that we need to consider only the
leaves of the subtree rooted at $v$.  Therefore, we show that a
leaf with the maximum weight appears as the winner of $v$.

We use induction on the height (the distance from the leaves) of the
nodes to show that the property holds for every node.
For leaves, the statement is trivially true.
Let $v$ be a node with children $x$ and $y$.
We denote by $T_w$ the subtree rooted at node $w$.
Then, every leaf of $T_v$ is a leaf in either $T_x$ or $T_y$.
By induction hypothesis, a leaf with the maximum weight in $T_x$
and $T_y$ appears as $\mathrm{winner}(x)$ and $\mathrm{winner}(y)$, respectively.
Therefore, the square with the maximum weight in $T_v$, $v'$, is
either $\mathrm{winner}(x)$ or $\mathrm{winner}(y)$.

Let $x' = \mathrm{winner}(x)$ and $y' = \mathrm{winner}(y)$
and let $p_x = x_1 x_2 ... x_{h_x}$
and $p_y = y_1 y_2 ... y_{h_y}$
be the path from the root of $T$ to $x'$ and $y'$, respectively
($x_1 = y_1$ is the root, $h_x$ is the depth of $x'$, $h_y$ is the
depth of $y'$, $x_{h_x} = x'$ and $y_{h_y} = y'$).
Both $p_x$ and $p_y$ include $v$; let $x_i = v$.
Since $p_x$ and $p_y$ diverge at $v$, $x_j = y_j$ for
every integer $j$ such that $1 \le j \le i$.
Based on Lemma~\ref{Lleafweight},
the weight of $x'$ is $\sum_{j=1}^{i} s_{x_j} + \sum_{j=i+1}^{h_x} s_{x_j}$
and the weight of $y'$ is $\sum_{j=1}^{i} s_{x_j} + \sum_{j=i+1}^{h_y} s_{y_j}$.
Also based on Lemma~\ref{Lleafwinner}, $f_x = \sum_{j=i+1}^{h_x} s_{x_j}$
and $f_y = \sum_{j=i+1}^{h_y} s_{y_j}$.
Therefore, $v'$ is $x'$ if $f_x > f_y$ and $y'$, otherwise.
This implies that $v'$ is the same as $\mathrm{winner}(v)$,
since $\mathrm{winner}(v)$ is chosen based on the value of $f_x$ and $f_y$.
This completes the proof.
\end{proof}

The main challenge in the analysis of the sweeping algorithm is limiting
the number of failure events.
In a dynamic and kinetic tournament tree for movement functions with
degree at most $s$,
using a balanced binary tree and when implementing each update as a deletion
followed by an insertion, the number of events is $O(\frac{m}{n}\lambda_{s+2}(n) \log n)$,
where $m$ is the number of updates and $\lambda_s(n)$ is the
maximum length of Davenport-Schinzel sequences of order $s$ on $n$ symbols \cite{agarwal08}.
For our problem, this yields a poor bound, since each update event may update
the weight of $O(n)$ leaves and thus $m = O(n^2)$,
which implies that the total number of failure events is $O(n^2\alpha(n) \log n)$,
in which $\alpha(n)$ denotes the inverse Ackermann function.
In Theorem~\ref{Tanalysis12} we present a tighter bound.

\begin{thm}
\label{Tanalysis12}
The time complexity of the plane sweep algorithm for finding a
hotspot of a horizontal trajectory with $n$ edges is $O(n \log^3 n)$
\end{thm}
\begin{proof}
After a failure event for a node and updating its winner and winner function,
we update its parent (unless it is the root).
At each update for node $v$, we check if the winner at $v$ needs to
be changed.  If so, we also update its parent.  Otherwise, we stop.
Therefore, instead of limiting the number of failure events,
we can find an upper bound on the total number of winner changes at
different nodes of $T$.

Let $\mathrm{winner}(v) = u$ at some point in the algorithm,
in which $u$ is a leaf.  Since weight functions are linear,
when $\mathrm{winner}(v)$ changes to a value $w$, where $w \ne u$,
$u$ can never become a winner at $v$, unless an update event updates
the weight of $w$ or $v$.
Without the update events, therefore, the number of times a leaf can become
a winner in its parent nodes is $O(\log n)$.  This implies a total
of $O(n \log n)$ winner changes.  It remains to limit the number of
winner changes that can result from update events.

Suppose an update event for edge $e$ updates the function assigned
to segment $g$.
Let $S$ be the set of all nodes like $v$ in $T$ such that $g \in I(v)$.
For every node $v$ in $S$ the sum and winner functions are updated.
This change does not cause any winner change in $T_v$ (the subtree
rooted at $v$),
because the relative weight of its leaves does not change.
However, the new winner function of $\mathrm{winner}(v)$
may cause $O(\log n)$ future winner changes in the ancestor of $v$.
In segment trees one can show that the label of each segment
appears in $O(\log n)$ nodes (for details, see \cite{berg08})
and thus the size of $S$ is $O(\log n)$.
Therefore, the number of winner changes by each update event
is $O(\log^2 n)$ and, since there are $O(n)$ update events,
the total number of winner changes induced by the update events
is $O(n \log^2 n)$.
Since for each winner change at node $v$, the failure time of $v$ is
updated in $Q$, the cost of performing each winner change is $O(\log n)$.
Thus, the time complexity of the algorithm is $O(n \log^3 n)$.
\end{proof}

\begin{thm}
\label{Tsweep12}
There exists an approximation algorithm with the approximation factor $1/2$
and time complexity $O(n \log^3 n)$ for finding a hotspot of
an orthogonal trajectory with $n$ edges.
\end{thm}
\begin{proof}
Follows from Theorem~\ref{Talg12}, Theorem~\ref{Tanalysis12}, and Lemma~\ref{hdecomp}.
\end{proof}

\section{Extension to Three Dimensions}
\label{scon}
Any algorithm used for finding hotspots in 2-dimensions can be extended to
find axis-parallel, cube hotspots of fixed side length for orthogonal
trajectories in $\mathbb{R}^3$.
We extend the definitions and notations presented in
Section~\ref{sprel} to $\mathbb{R}^3$.
The weight of a cube $c$ with side length $s$ with respect
to trajectory $T$ in $\mathbb{R}^3$ is the total duration in which the
entity spends inside it; we represent it as $w_T (c)$, as before.
A hotspot of a trajectory $T$ in $\mathbb{R}^3$ is an axis-parallel
cube (i.e.~a cube whose faces are parallel to the planes defined by
any pair of the axes of the coordinate system) of fixed
side length $s$ and the maximum weight, $h(T)$.

Let $e$ be an edge parallel to the $z$-axis and let $c$ be an
axis-parallel cube.  Exactly two faces of $c$ are parallel to the $xy$-plane,
$Z_1(c)$ and $Z_2(c)$, with $Z_1(c)$ appearing first (in the positive direction of
the $z$-axis).  The contribution rate of $e$ with respect to
$c$, denoted as $c_z(e)$, is the rate at which the contribution of the
weight of $e$ to the weight of $c$ increases if $c$ is moved in the
positive direction of the $z$-axis.
As in the 2-dimensional case, the absolute value of $c_z(e)$ is
either zero or the ratio of the duration of $e$ to its length,
which we denote as $d(e)$.
We define $c_z(T)$ for orthogonal trajectory $T$ as the sum of the
contribution rates of all edges of $T$ that are parallel to the $z$-axis.
The following lemma extends Lemma~\ref{hside} to three dimensions.

\begin{lem}
\label{tside3d}
Let $T$ be a trajectory in $\mathbb{R}^3$ with axis-parallel edges.
For any axis-parallel cube like $c$, there is a cube with at
least the same weight,
such that a vertex of $T$ is on one of the two planes formed
by extending its $xy$-parallel faces.
\end{lem}

Therefore, to find a hotspot of $T$, it suffices to search among the cubes
with a vertex of $T$ on one of the $xy$-parallel planes containing its
$xy$-parallel faces.  This observation suggests Theorem~\ref{T3d}.

\begin{thm}
\label{T3d}
Suppose algorithm $A$ can find a $c$-approximate hotspot of
any trajectory in $\mathbb{R}^2$ containing $n$ axis-aligned edges
with the time complexity $O(t(n))$.
For a trajectory $T$ in $\mathbb{R}^3$, all of whose $n$ edges are axis-aligned,
there exists an algorithm with the time complexity $O(n \cdot t(n))$
and approximation factor $c$ to find an axis-aligned cube hotspot of $T$.
\end{thm}
\begin{proof}
For each vertex $v$ of $T$, let $z(v)$ be the $z$-coordinate of $v$.
Project all edges that are (maybe partially) between $z = z(v)$ and $z = z(v) + s$
to the plane $z = z(v)$ to obtain an orthogonal 2-dimensional trajectory $T'$.
Edges parallel to the $z$-axis are projected to an edge with length zero,
whose weight denotes the duration of the portion between $z = z(v)$ and $z = z(v) + s$.
Perform algorithm $A$ on $T'$ to obtain a square $s$.
Let $c$ be the cube with $Z_1(c)$ on $s$.
It is not difficult to see that $w_T (c)$ is equal to $w_{T'}(s)$.
Record $c$, if it has the maximum weight so far.
Repeat the preceding steps after reversing the direction of the $z$-axis
to find cubes like $c$, with $Z_2(c)$ on a vertex of $T$.
Return the cube with the maximum weight.
Lemma~\ref{tside3d} implies that this cube is a $c$-approximate
hotspot of $T$.
\end{proof}

Theorem~\ref{T3d} and Theorem~\ref{Tsweep12} imply Corollary~\ref{Tcube12}.
\begin{cor}
\label{Tcube12}
A $1/2$-approximate cube hotspot of a three-dimensional trajectory with
$n$ edges can be found with the time complexity $O(n^2 \log^3 n)$.
\end{cor}

It seems possible to generalize the result of Corollary~\ref{Tcube12}
to $R^d$ to obtain an algorithm with the time complexity $O(n^{d-1} \log^3 n)$.

\section*{Acknowledgements}
We are grateful to the anonymous reviewers for several suggestions
to formalize the concepts of the paper and improve its presentation.
We also wish to thank Marc van Kreveld for his valuable
comments on an earlier version of this paper.


\end{document}